\documentclass{article}
\usepackage[utf8]{inputenc}
\usepackage{graphicx}
\graphicspath{{./pictures/}}
\usepackage{amsmath,amsfonts,amsthm}

\usepackage{color}
\usepackage{url}

\allowdisplaybreaks

\newcommand{\R}{\mathbb{R}}
\newcommand{\C}{\mathbb{C}}

\newtheorem{Theorem}{Theorem}

\title{A quasi separable 
dissipative Maxwell-Bloch  system for laser dynamics}
\author{Gianluca Gorni\\
Universit\`a di Udine\\
Dipartimento di Matematica e Informatica\\
via delle Scienze~208, 33100 Udine, Italy\\
\tt{gianluca.gorni@uniud.it}
\and
Stefania Residori\\
Institut de Physique de Nice, UMR7010 \\ 
Universit\'e de Nice - Sophia Antipolis, CNRS\\
1361 route des Lucioles, 06560 Valbonne, France\\
\tt{stefania.residori@inphyni.cnrs.fr}
\and
Gaetano  Zampieri\\
Universit\`a di Verona\\
Dipartimento di Informatica\\
strada Le Grazie 15, 37134 Verona, Italy\\
\tt{gaetano.zampieri@univr.it}}

\date{12 August 2017}

\begin{document}

\maketitle
\begin{abstract}
The Maxwell-Bloch dissipative equations describe laser dynamics. Under a simple condition on the parameters there exist two time dependent first integrals, that  allow a nonstandard separation of variables in the equations. That condition has a precise physical meaning. The separated differential equations lead naturally to simple conjectures on the asymptotic behavior of the physical variables.
\end{abstract}

Subject classes: 78A60, 37J15.

\section*{keywords}
Dissipative Maxwell-Bloch; nonstandard separation of variables.

\section{Introduction}
\label{Introduction}

The Maxwell-Bloch equations are well-know to describe laser dynamics for a system of two-level atoms in a cavity resonator. They were first derived in a 1965 paper by Arecchi et Bonifacio \cite{Arecchi_1965}. They have come to be credited as Maxwell-Bloch equations because of coupling the Maxwell equations for the envelope of the electric field with the Bloch description of atomic electric dipoles in interactions with electromagnetic fields (see, e.g., \cite{Nature_2015}). As a matter of fact, the Maxwell-Bloch equations are widely employed as a prototype model of light-matter interaction \cite{Newell}. While their derivation uses a quantum (semi-classical) approach to polarization of a two-level atom and density matrix, the resulting equations are classical and describe the macroscopic (classical) electric field amplitude, polarization, and population inversion.

Following Arecchi et al.~\cite{ArecchiMeucci}, \cite{Arecchi}, \cite{ArecchiTredicce}, the Maxwell-Bloch equations can be written as
\begin{equation}\label{MaxwellBlochArecchiMeucci}
  \begin{split}
  \dot E={}&-(\kappa+i\zeta)E+gP,\\
  \dot P={}&-(\gamma_\perp+i\delta)P+gE\Delta,\\
  \dot\Delta={}&-\gamma_\parallel(\Delta-\Delta_0)-
  2g(E^*P+EP^*)
  \end{split}
\end{equation}
where
\begin{itemize}
\item $E$ is the complex amplitude of the electric field,
\item $P$~the complex polarization of the atomic medium,
\item $\Delta$ is the real population inversion,
\item $\kappa$, $\gamma_\perp$, $\gamma_\parallel>0$ are the loss rates of $E,P,\Delta$ respectively,
\item $g>0$ is the coupling constant,
\item $\Delta_0$ is the real equilibrium population inversion in the absence of the field,
\item $\zeta$ is the cavity mistuning, 
\item $\delta$ is the detuning between the field frequency and the center of the atomic line.
\end{itemize}
The star in $E^*,P^*$ indicates the complex conjugate. For a more detailed discussion on the physical meaning of the variables and parameters we refer to the original articles. 

Laser are usually classified into classes according to the relative relaxation times for the electric field, polarization and population inversion and consequent adiabatic elimination of the fast variables~\cite{ArecchiMeucci}. In some situations Maxwell-Bloch equations exhibit chaotic behaviors \cite{Haken_1975}, \cite{Arecchi}, \cite{ArecchiTredicce}, when at least three variables are involved in the dynamics. For comparable relaxation times all the three equations must be kept and it can be shown that Maxwell-Bloch equations are  equivalent to the Lorenz equations \cite{Haken_1975}. When adiabatic elimination for the polarization and/or the population reduces the number of equations, laser chaotic behaviors can be produced by adding another variable, for instance by modulating losses \cite{Arecchi} or by injecting an external signal in the cavity \cite{ArecchiTredicce}.

In a previous paper of ours~\cite{GZkilling} we introduced a nonvariational Lagrangian setting for Maxwell-Bloch equations in the dissipative limit, that is, when both the cavity mistuning and the detuning of the electric field frequency from the center of the atomic line can be set to zero:
\begin{equation}
  \zeta=0,\qquad \delta=0.
\end{equation}
In this limit we derived an interesting nonlocal constant of motion, which becomes a true time-dependent first integral for the particular case 
 when the loss rate of the population inversion is  two times the loss rate of the photons in the cavity:
\begin{equation}\label{separableCase}
 \gamma_\parallel=2\kappa.
\end{equation}
In Section~\ref{DissipativeMBsection} below we will provide a simplified, self-contained derivation of the first integral.

The main result of this work is to show in Section~\ref{separation} that this first integral, together with another one which is obvious, leads to a certain kind of separation of variables: if we write $E$ in polar form $E=re^{i\theta}/2$ (the $/2$ simplifies equations \eqref{MB-5-D} below consistently with previous papers), we can write a second-order nonautonomous differential equation~\eqref{radialeq} for $r(t)$
\begin{equation}\label{separataRPhysical}
  \ddot r=-(\kappa+\gamma_\perp)\dot r+
  \Bigl(g^2\Delta_0-\kappa\gamma_\perp+
  \frac{g^2M}{2} e^{-2\kappa t}\Bigr)r
  -\frac{g^2}{2}r^3+
  \frac{N^2}{r^3}e^{-2(\kappa+\gamma_\perp)t},
\end{equation}
containing the first integral values $N,M$ but not $\theta,P,\Delta$, which can be easily derived from~$r$ by quadrature.

The final form of equation~\eqref{separataRPhysical} and some numerical experiment lead us to conjecture that the system has a very simple asymptotic behavior:
\begin{itemize}
\item if $g^2\Delta_0\le \kappa \gamma_\perp$ then $E,P,\Delta$ all vanish as $t\to+\infty$; Figure~\ref{zeroAsymptotics} shows such a trajectory of~$E$ and~$P$ in the complex plane, computed numerically.

\item If $g^2\Delta_0> \kappa \gamma_\perp$ then $E$ converges to a constant $E_\infty\in \C$ such that $\lvert E_\infty \rvert^2 =( g^2\Delta_0- \kappa \gamma_\perp)/(2g^2)$, $P(t)\to \kappa E_\infty/g$, and $\Delta\to\kappa \gamma_\perp$, as $t\to+\infty$; Figure~\ref{asymptoticCircle1} is a sample of this situation.
\end{itemize}

\begin{figure}\centering
\includegraphics[width=\textwidth]{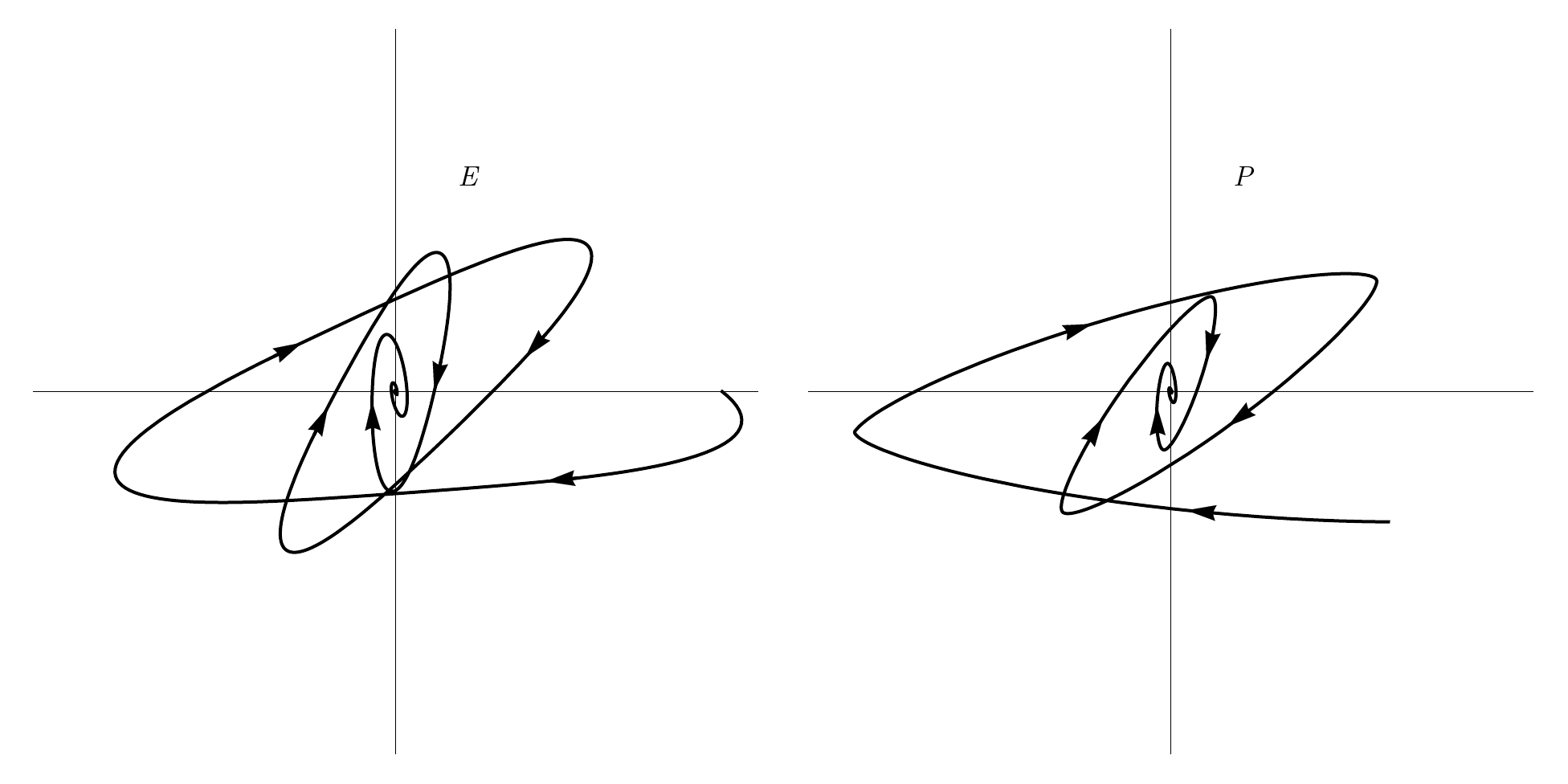}
\caption{A forward orbit of $E$ and $P$ in the complex plane in a case $g^2\Delta_0\le \kappa \gamma_\perp$: they converge to~0.}
\label{zeroAsymptotics}

\bigskip

\includegraphics[width=\textwidth]{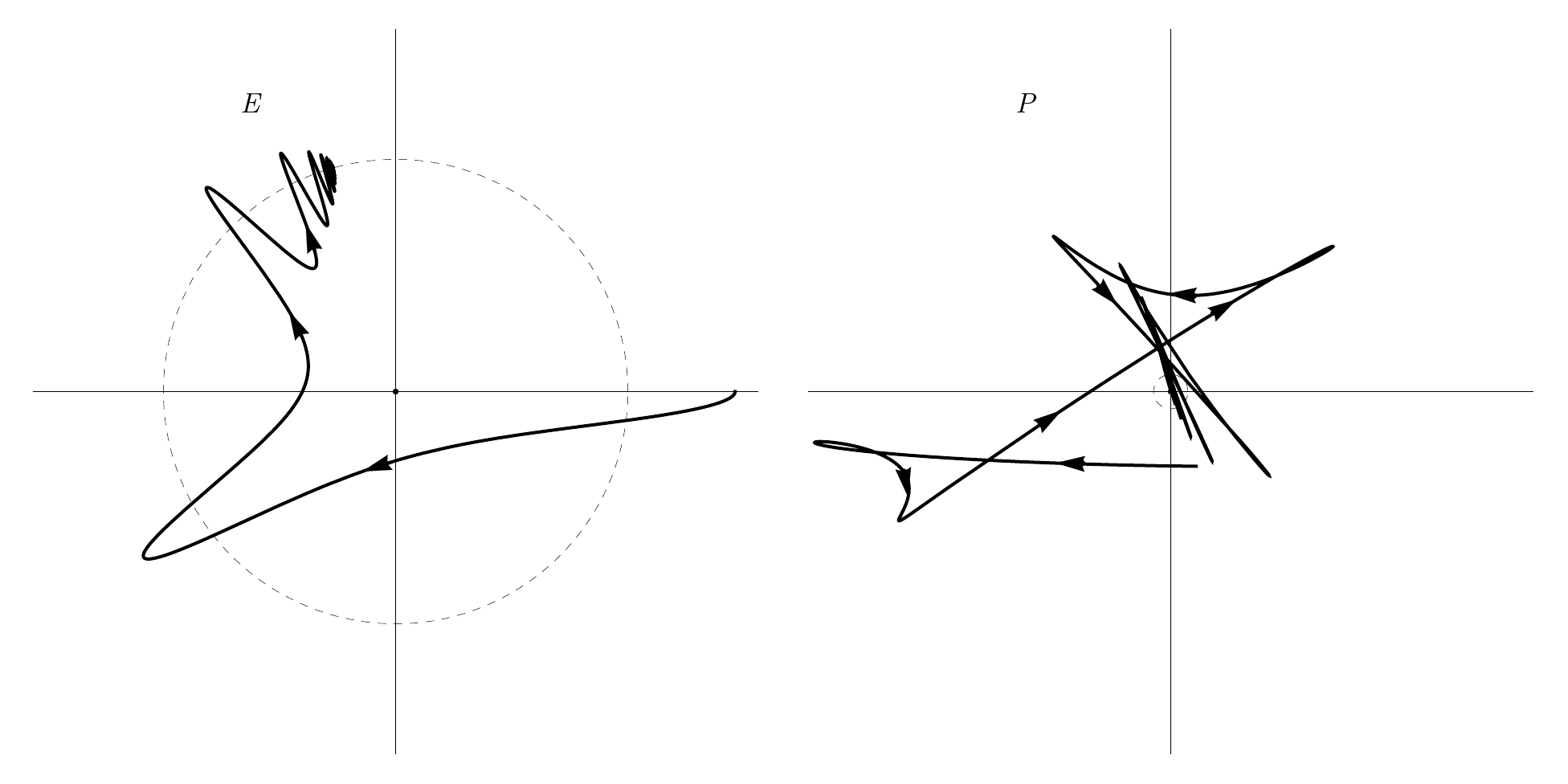}
\caption{A forward orbit of $E$ and of~$P=\kappa E+\dot E$ in the complex plane in a case $g^2\Delta_0> \kappa \gamma_\perp$: $E$~oscillates around, and converges to, a point of an asymptotic circle of radius $\lvert E_\infty \rvert=\sqrt{g^2 \Delta_0- \kappa \gamma_\perp}/(\sqrt{2}g)$; $P$~converges to a point on the (small in the picture) circle of radius $\kappa \lvert E_\infty \rvert/g$.} \label{asymptoticCircle1}
\end{figure}


\section{Lagrangian setting for the dissipative Maxwell-Bloch equations}\label{DissipativeMBsection}

In equations~\eqref{MaxwellBlochArecchiMeucci} we introduce the new real variables $x_1,x_2, y_1,y_2,z$ as
\begin{equation}
  x_1+ix_2=2E,\quad
  y_1+iy_2=2P,\quad
  z=\Delta,
\end{equation}
and restrict us to the dissipative limit $\theta=0$, $\delta=0$. To be consistent with our previous work~\cite{GZkilling}, we rename the coefficients this way
\begin{equation}
  a=\kappa,\quad b=\gamma_\perp,\quad c=\gamma_\parallel,
  \quad k=\Delta_0.
\end{equation}
We obtain what we will call \emph{the dissipative Maxwell-Bloch equations}, which are the following 5-dimensional system with $a,b,c,g>0$, and $k\in\mathbb{R}$ parameters
\begin{equation}\label{MB-5-D}
  \begin{cases}
  \dot x_1=-ax_1+gy_1,\\
  \dot x_2=-ax_2+gy_2,\\
  \dot y_1=-by_1+gx_1 z,\\
  \dot y_2=-by_2+gx_2z,\\
  \dot z=-c\bigl(z-k\bigr)-g(x_1y_1+x_2y_2).
  \end{cases}
\end{equation}

To rewrite~\eqref{MB-5-D} in a nonvariational Lagrangian setting, we choose the Lagrangian variables as
\begin{equation}
  q_1=x_1,\quad q_2=x_2,\quad \dot q_3=z.
\end{equation}
This choice gives $\dot q_1=\dot x_1=-ax_1+gy_1$ and $\dot q_2=\dot x_2=-ax_2+gy_2$, so
\begin{equation}\label{velocity}
  y_1=\frac{\dot q_1+a q_1}{g},
  \qquad
  y_2=\frac{\dot q_2+a q_2}{g}.
\end{equation}
The equations~\eqref{MB-5-D} become
\begin{equation}\label{ELDissipativeConcrete}
  \begin{cases}
  \ddot q_1=-abq_1-(a+b)\dot q_1+g^2q_1\dot q_3\\
  \ddot q_2=-abq_2-(a+b)\dot q_2+g^2q_2\dot q_3\\
  \ddot q_3=-a\bigl(q_1^2+q_2^2\bigr)
  -c(\dot q_3-k)
  -\bigl(q_1\dot q_1+q_2\dot q_2\bigr)
  \end{cases}
\end{equation} 
We set the vector variable $q=(q_1,q_2,q_3)$ and introduce the Lagrangian formulation as
\begin{gather}\label{Ldissipative}
  L(t,q,\dot q)=\frac{1}{2}
  \bigl(\dot q_1^2+\dot q_2^2+g^2\dot q_3^2+(q_1^2+q_2^2)
  (g^2\dot q_3-ab\bigr)\bigr), \\
  \label{Qdissipative}
  Q(t,q,\dot q)=\Bigl(-(a+b)\dot q_1,
  -(a+b)\dot q_2,
  -ag^2\bigl(q_1^2+q_2^2\bigr)-cg^2(\dot q_3-k)\Bigr).
\end{gather}
The equations~\eqref{ELDissipativeConcrete} become the nonvariational Lagrange equation
\begin{equation}\label{nonVariationalLagrange}
  \frac{d}{dt}\partial_{\dot q}L\bigl(t,q(t),\dot q(t)\bigr)
  -\partial_{q}L\bigl(t,q(t),\dot q(t)\bigr)=
  Q\bigl(t,q(t),\dot q(t)\bigr),
\end{equation}
as can be checked easily. For systems described by this kind of equation we have developed a theory of nonlocal constants of motion~\cite{GZkilling}. 

The only result we need here is the following:

\begin{Theorem}\label{theoremOfNonvariationalConstantOfMotion}
Let  $L(t,q,\dot q)$, $q,\dot q\in\R^n$, be a smooth Lagrangian function, let $t\mapsto q(t)$ be a solution to the (nonvariational) Lagrange equation~\eqref{nonVariationalLagrange} and let $q_\lambda(t)$ be a  family of  perturbed motions, smooth in  $(\lambda,t)$, with  $\lambda$ in a neighbourhood of $0\in\R$, and such that $q_\lambda(t)\equiv q(t)$ when $\lambda=0$.
Then the following function is constant:
\begin{multline}\label{ConstantAlongNonvariationalMotion}
  t\mapsto
  \partial_{\dot q}
  L\bigl(t,q(t),\dot q(t)\bigr)\cdot
  \partial_\lambda q_\lambda(t)
  \big|_{\lambda=0}-
  \int_{t_0}^t\biggl(
  \frac{\partial}{\partial\lambda}
  L\bigl(s,q_\lambda(s),\dot q_\lambda(s)\bigr)
  \big|_{\lambda=0}+\\
  +Q\bigl(s,q(s),\dot q(s)\bigr)\cdot
  \partial_\lambda q_\lambda(s)
  \big|_{\lambda=0}\biggr)ds\,.
\end{multline}
\end{Theorem}

\begin{proof}
Simply take the time derivative of~\eqref{ConstantAlongNonvariationalMotion}, reverse the derivation order and use the Lagrange equation~\eqref{nonVariationalLagrange}.
\end{proof}

We introduced the variational case ($Q\equiv 0$) of the theorem in our 2014 paper~\cite[Th.~3]{GorniZampieri}, and provided more applications later~\cite{GZnonlocal}. In particular, we devoted a paper~\cite{GZMB} to the Maxell-Bloch equation in the conservative case ($a=b=c=0$), giving a separation of variables on fixed levels of the energy and other first integrals, in a way related to the theory in Pucacco and Rosquist~\cite{Pucacco}.

We have already found two choices of perturbed motion $q_\lambda(t)$ that lead to interesting nonlocal constants of motion.

Consider first the rotation
\begin{equation}
  q_\lambda(t)=
  \begin{pmatrix}\cos\lambda&
  \sin\lambda&0\\
  -\sin\lambda&\cos\lambda&0\\
  0&0&1
  \end{pmatrix}
  \begin{pmatrix}q_1(t)\\q_2(t)\\q_3(t)\end{pmatrix},\quad
   \partial_\lambda q_\lambda(t)
  \big|_{\lambda=0}=\bigl(-q_2(t), q_1(t),0\bigr),
\end{equation}
and plug this into formula~\eqref{ConstantAlongNonvariationalMotion} of Theorem~\ref{theoremOfNonvariationalConstantOfMotion}. We obtain this constant of motion:
\begin{equation}
  q_1(t)\dot q_2(t)-q_2(t)\dot q_1(t)
  +(a+b)\int_{t_0}^t
  \bigl(q_1(s)\dot q_2(s)-q_2(s)\dot q_1(s)\bigr)ds,
\end{equation}
that is, $J+(a+b)\int J\,ds$ where $J=q_1\dot q_2-q_2\dot q_1$ is the angular momentum. So $\dot J+(a+b)J=0$ and
we have the following time-dependent first integral
\begin{equation}\label{firstintegralangmom} 
  N=e^{(a+b)t} \bigl(q_1\dot q_2-q_2\dot q_1\bigr).
\end{equation}

The next interesting choice is $ q_\lambda(t)= q(t)+ \lambda(0,0, 2e^{ct})$, for which the constant of motion is
\begin{multline}
  g^2e^{ct}\bigl(q_1(t)^2+q_2(t)^2+2\dot q_3(t)\bigr)-
  g^2\int e^{ct}\Bigl(2ck +
  (c- 2a)\bigl(q_1(t)^2 + q_2(t)^2\bigr)
  \Bigr)dt,
\end{multline}
which, after a simple integration by part, becomes
\begin{multline}
  g^2e^{ct}\bigl(q_1(t)^2+q_2(t)^2+2\dot q_3(t)-2k\bigr)+
  (2a-c)g^2\int e^{ct}\bigl(q_1(t)^2+q_2(t)^2\bigr)dt.
\end{multline}
We deduce that the quantity
\begin{equation}\label{secondMonotonicForMaxwellBloch}
  t\mapsto
  e^{ct}\bigl(q_1(t)^2+q_2(t)^2+2\dot q_3(t)-2k\bigr) 
\end{equation}
is monotonic, and it is a true first integral when $c=2a$:
\begin{equation}\label{firstintc=2a}
  M=e^{2at}\bigl(q_1(t)^2+q_2(t)^2+2\dot q_3(t)-2k
  \bigr).
\end{equation}
In the sequel we will assume the case $c=2a$, which translates as~\eqref{separableCase} in the original physical notation.


\section{Quasi separation}\label{separation}

As first step we solve the conservation law \eqref{firstintc=2a} for $\dot q_3$ and we plug the result into the first two equations of~\eqref{ELDissipativeConcrete}, obtaining nonautonomous equations for $q_1,q_2$ that do not depend on~$q_3$:
\begin{equation}\label{eqsq1q2casoc=2a}
  \begin{split}
  \ddot q_1&=-abq_1-(a+b)\dot q_1+g^2
  \Bigl(k+\frac{1}{2} M e^{-2at}-\frac{1}{2}\bigl( q_1^2+q_2^2\bigr)
  \Bigr)q_1\\
  \ddot q_2&=-abq_2-(a+b)\dot q_2+g^2
  \Bigl(k+\frac{1}{2} M e^{-2at}-\frac{1}{2}\bigl( q_1^2+q_2^2\bigr)
  \Bigr)q_2.
  \end{split}
\end{equation} 
The point $(q_1,q_2,\dot q_1,\dot q_2)=0$ is an equilibrium position for all values of the parameter $M$. The corresponding solutions of the full initial system are deduced by \eqref{firstintc=2a}: $\dot q_3(t)=k+M e^{-2at}/2$. 

If we introduce polar coordinates $(r,\theta)$ the $(q_1,q_2)$ plane, we will obtain a nonautonomous equation for~$r$ that does not contain~$\theta$.

If we multiply equations~\eqref{eqsq1q2casoc=2a}, multiplied respectively by $\cos\theta$ and $\sin\theta$ and summed term by term, 
and using the formula for the radial acceleration
\begin{equation}
 \ddot q_1 \cos\theta+ \ddot q_2\sin\theta=\ddot r-r\dot \theta^2,
\end{equation}
we obtain
\begin{align*}
  \ddot r-r\dot \theta^2={}&
  -abr\cos^2\theta-(a+b)
  (\dot r\cos\theta-r\dot\theta\sin\theta)\cos\theta+{}\\
  &{}+g^2\Bigl(k+\frac{1}{2} M e^{-2at}
  -\frac{1}{2}r^2\Bigr)r\cos^2\theta
  -abr\sin^2\theta+{}\\
  &{}+
  g^2\Bigl(k+\frac{1}{2} M e^{-2at}-
  \frac{1}{2}r^2\Bigr)r\sin^2\theta-{}\\
  &-(a+b)
  (\dot   r\sin\theta+r\dot\theta\cos\theta)\sin\theta=\\
  ={}&-abr-(a+b)\dot r
  +g^2\Bigl(k+\frac{1}{2} M e^{-2at}-\frac{1}{2}r^2\Bigr)r.
\end{align*}
We can eliminate $\dot\theta$ using the first integral \begin{equation}\label{arealintegral}
  r^2\dot\theta=N e^{-(a+b)t}.
\end{equation}
Rearranging we arrive at our separated nonautonomous equation for the radial coordinate~$r$:
\begin{equation}\label{separataR}
  \ddot r=-(a+b)\dot r+\Bigl(g^2k-ab+\frac{g^2M}{2} e^{-2at}\Bigr)r
  -\frac{g^2}{2}r^3+\frac{N^2}{r^3}e^{-2(a+b)t}.
\end{equation}
It can be put into the variational Lagrangian framework
\begin{equation}\label{radialeq}
 \frac{d}{dt}\partial_{\dot r}{\cal L}\bigl(t,r(t),\dot r(t)\bigr)
  -\partial_{r}{\cal L}\bigl(t,r(t),\dot r(t)\bigr)=0
\end{equation}
with
\begin{equation}\label{radialLagrangian}
  {\cal L}=e^{(a+b)t}
  \biggl(\frac{1}{2} \dot r^2+\Bigl(g^2k-ab+
  \frac{g^2M}{2}
  e^{-2at}\Bigr)\frac{r^2}{2}-\frac{g^2}{8}r^4-
  \frac{N^2}{2r^2}e^{-2(a+b)t}\biggr).
\end{equation}

\section{Conjectures on the asymptotic behavior}\label{asymptoticBehavior}

To get some feeling of the asymptotic behavior of~$r(t)$ as $t\to+\infty$, let us see what happens if the time exponentials $e^{-2at}$ and $e^{-2(a+b)t}$ in equation~\eqref{separataR} are replaced by their limit~0:
\begin{equation}\label{separataRasintotica}
  \ddot r=-(a+b)\dot r+(g^2k-ab)r
  -\frac{g^2}{2}r^3.
\end{equation}
This limiting equation has constant solutions corresponding to the nonnegative solutions of the algebraic equation
\begin{equation}\label{equilibriumRadiusEquation}
  (g^2k-ab)r-\frac{g^2}{2}r^3=0.
\end{equation}
There are clearly two cases:

\begin{itemize}

\item If $g^2k\le ab$ then equation~\eqref{equilibriumRadiusEquation} has only the solution $r=0$, and all solutions of the simplified equation~\eqref{separataRasintotica} converge to $r=0$ as $t\to+\infty$. We suspect that the same happens for the original equation~\eqref{separataR}, whose last term $+N^2e^{-2(a+b)t}/r^3$ will repel the solutions from the origin and avoid that the singularity $r=0$ be reached in finite time. Figure~\ref{zeroAsymptotics} shows a typical trajectory on the $E=(q_1+iq_2)/2$ plane.

\item If $g^2k>ab$ we have two nonnegative solutions $r=0$ and
\begin{equation}\label{asymtpticRadius}
  r_\infty=\frac{1}{g}\sqrt{2(g^2k-ab)}.
\end{equation}
The positive solutions of the simplified equation~\eqref{separataRasintotica} all converge to~$r_\infty$ as $t\to+\infty$. We conjecture that also the solutions to the original equation~\eqref{separataR} converge to the same limit. Equation~\eqref{arealintegral} will mean an exponential decay for $\dot\theta$, and therefore a finite limit for $\theta(t)$ as~$t\to+\infty$. The conjecture is sustained by numeric simulations. Figure~\ref{asymptoticCircle1} shows a typical trajectory on the $E=(q_1+iq_2)/2$ plane, with the asymptotic circle shown dashed. Figure~\ref{asymptoticCircle2} is a variant, where we have chosen a starting point on the asymptotic circle with null radial speed.

\end{itemize}

In both cases we conjecture that  $E(t)=\bigl(q_1(t)+iq_2(t)\bigr)/2$ converges to finite limit~$E_\infty$ as $t\to +\infty$ while $\dot E(t)$ vanishes in the limit $t\to +\infty$. Also this conjecture is sustained by numeric simulations. If this is true, then $P=(aE+\dot E)/g$ converges to $aE_\infty/g$.

\begin{figure}\centering
\includegraphics[width=\textwidth]{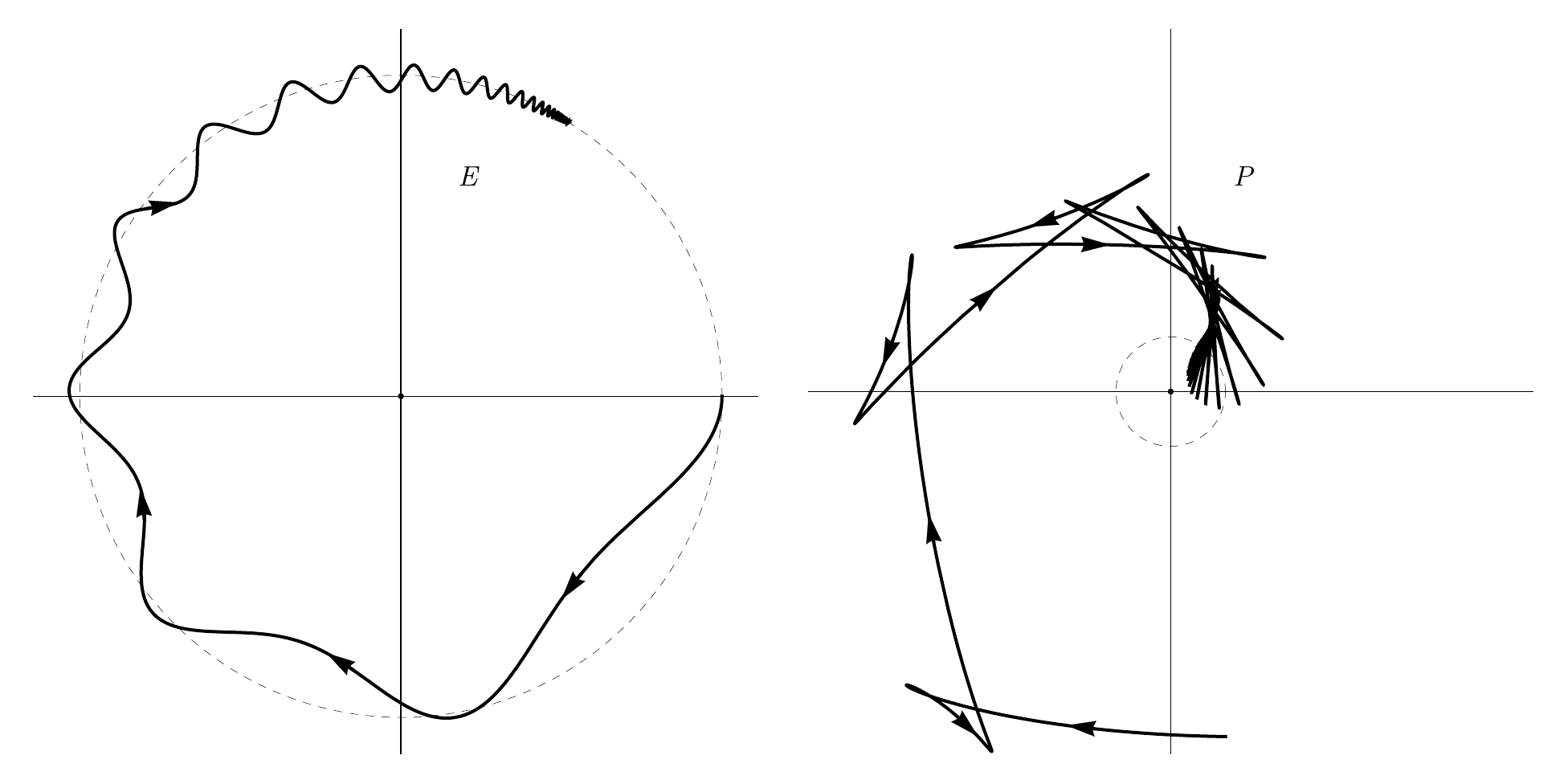}
\caption{A forward orbit of $E$ and $P$ in the complex plane in a case when $g^2\Delta_0> \kappa \gamma_\perp$, with $E$ starting at a point of the asymptotic circle  with null radial speed.}
\label{asymptoticCircle2}
\end{figure}

\section*{Acknowledgment}
The research was done under the auspices of INdAM (Istituto Nazionale di Alta Matema\-tica). 


\end{document}